\documentclass[a4paper]{llncs}
\usepackage{wrapfig}
\usepackage{xspace}
\usepackage[counterclockwise, figuresleft]{rotating}
\usepackage{amssymb}
\usepackage{amsmath}
\usepackage{stmaryrd}
\usepackage[usenames,dvipsnames]{xcolor}
\usepackage{proof-dashed}
\usepackage{tikz}
\usetikzlibrary{arrows}
\usetikzlibrary{positioning}
\usepackage{todonotes}
\usepackage{extarrows}
\usepackage{floatrow}

\newcommand{\csomit}[1]{}

\newcommand{\co}[1]{\overline{#1}}

\usepackage{colortbl}

\newcommand{\Axcut}{\m{C}_{\m{Ax}}}
\newcommand{\onebot}{\m{C}_{\one\bot}}
\newcommand{\tensorparr}{\m{C}_{\tensor\parr}}
\newcommand{\opluswith}{\m{C}_{\oplus\with}}
\newcommand{\bangwhynot}{\m{C}{!?}}
\newcommand{\bangw}{\m{C}_{!\m{w}}}
\newcommand{\bangC}{\m{C}_{!\m{C}}}
\newcommand{\existsforall}{\m{C}_{\exists\forall}}

\tikzstyle{wto} = [->>]

\newcommand{\til}[1]{\tilde{#1}}
\newcommand{\wtil}[1]{\widetilde{#1}}

\newcommand{\send}[5]{\co{#1} (#2);(#4\pp #5)}
\newcommand{\gensend}{\send{x}{x'}{y',\til z'}{P}{Q}}

\newcommand{\wait}[2]{\m{wait}[#1];#2}
\newcommand{\genwait}{\wait{y}P}
\newcommand{\close}[1]{\m{close}[#1]}
\newcommand{\genclose}{\close{x}}

\newcommand{\caxcut}[4]{{#1}\leftarrow{#2}\rightarrow{#3}; #4}
\newcommand{\genaxcut}{\caxcut{w}{y^B}{x}{P}}

\newcommand{\recv}[4]{#1(#3);#4}
\newcommand{\genrecv}{\recv{y}{\til x'}{y'}P}
\newcommand{\Case}[3]{#1.\mathsf{case}(#2,#3)}
\newcommand{\gencase}{\Case{y}PQ}
\newcommand{\nocase}[3]{{#1}^{#3}.\m{case}()}
\newcommand{\gennocase}{\nocase{x}{y}{\til u}}
\newcommand{\inl}[2]{#1.\mathsf{inl};#2}
\newcommand{\geninl}{\inl{x}P}
\newcommand{\inr}[2]{#1.\mathsf{inr};#2}
\newcommand{\geninr}{\inr{x}Q}
\newcommand{\genginl}{\comsel{x}{\til y}{\lleft}{P}{Q_1,\dots,Q_n}}
\newcommand{\genginr}{\comsel{x}{\til y}{\lright}{P_1,\dots,P_n}{Q}}

\newcommand{\nil}{\boldsymbol 0}
\newcommand{\res}[1]{(\boldsymbol\nu #1)\,}
\newcommand{\gcpres}[2]{\res{#1 : #2}}

\newcommand{\inlinr}[3]{#1.(\lleft: #2, \lright: #3)}
\newcommand{\chinlinr}[4]{{#1}\mathrel{\rightarrow}{#2}.(\lleft: #3, \lright: #4)}
\newcommand{\geninlinr}{\inlinr{x}{P}{Q}}
\newcommand{\genchinlinr}{\chinlinr{x}{\til y}{P}{Q}}

\newcommand{\srv}[3]{\m{srv}\, {#1};{#3}}
\newcommand{\gensrv}{\srv yxP}
\newcommand{\client}[3]{\m{use}\, {#1};{#3}}
\newcommand{\genclient}{\client x{\til y}P}
\newcommand{\shutdown}[3]{\m{kill}\, {#1} \pp {#3}}
\newcommand{\srvkill}[1]{\m{kill}\, {#1}}
\newcommand{\genshutdown}{\shutdown x{\til y}P}
\newcommand{\clone}[5]{\m{clone}\, {#1}(#3); #5}
\newcommand{\genclone}{\clone x{\til y}{x'}{\til y'}P}
\newcommand{\cloneo}[2]{\m{clone}\, {#1}(#2)}

\newcommand{\gclient}[3]{{#1}\, \m{starts}\, {#2};{#3}}
\newcommand{\gengclient}{\gclient x{\til y}P}
\newcommand{\gshutdown}[3]{{#1}\, \m{kills}\, {#2};{#3}}
\newcommand{\gengshutdown}{\gshutdown x{\wtil{y(Q)}}P}
\newcommand{\gclone}[5]{{#1}\, \m{clones}\, {#2}(#3,#4);#5}

\newcommand{\gengclone}{\gclone x{\til y}{x'}{\til y'}P}
\newcommand{\sendtype}[4]{{#1}[{#3}];#4}
\newcommand{\gensendt}{\sendtype x{\til y}AP}
\newcommand{\recvtype}[4]{{#1}({#3});#4}

\newcommand{\comtype}[5]{#1[#3]\mathrel{\rightarrow}#2(#4);#5}
\newcommand{\gencomt}{\comtype x{\til y}AXP}
\newcommand{\markc}[1]{\underline{\textcolor{dkgreen}{\ensuremath{#1}}}}
\newcommand{\marka}[1]{\fbox{\textcolor{dkred}{\ensuremath{#1}}}}
\newcommand{\markb}\hl

\definecolor{dkred}{rgb}{.6,0,0}
\definecolor{dkgreen}{rgb}{0,.6,0}
\definecolor{dkblue}{rgb}{0,0,1}

\newcommand{\com}[5]{#1(#2)\mathrel{\rightarrow}#3(#4);#5}
\newcommand{\gencom}{\com{\til x}{\til x'}{y}{y'}{P}}
\newcommand{\comclose}[3]{#1\ \m{closes}\ #2;#3}
\newcommand{\gencomclose}{\comclose{\til x}yP}
\newcommand{\tensor}{\otimes}
\newcommand{\tensort}[1]{\mathbin{\tensor}}
\newcommand{\parr}{\,\bindnasrepma\,}
\newcommand{\parrt}[1]{\mathbin{\parr}}
\newcommand{\one}{{\bf 1}}
\newcommand{\onet}[1]{\one}
\newcommand{\bott}[1]{\bot}
\newcommand{\with}{{\,\&\,}}
\newcommand{\oplust}[1]{\mathbin{\oplus}}
\newcommand{\witht}[1]{\mathbin{\with}}
\newcommand{\bangt}[1]{{\bang}}
\newcommand{\whynott}[1]{\quest}
\newcommand{\existst}[1]{\exists}
\newcommand{\forallt}[1]{\forall}
\newcommand{\zerot}[1]{0}
\newcommand{\topt}[1]{\top}

\newcommand{\seq}{{\vdash}}

\newcommand{\cseq}{\mathrel{
\scalebox{0.8}{\raisebox{-0.5mm}{$\stackrel{\circ}{\scalebox{0.8}{\makebox[0mm][c]{$\circ$}}}$}}}
}
\newcommand{\gseq}{\mathrel{\vDash}}
\newcommand{\pp}{\ \boldsymbol{|}\ }
\newcommand{\pair}[2]{#1\!:\!#2}

\newcommand{\m}[1]{\ensuremath{\mathsf{#1}}}

\definecolor{light-gray}{gray}{0.948}
\newcommand{\hl}[1]{\colorbox{light-gray}{$#1$}}

\definecolor{cool}{gray}{0.95}

\newcommand{\judge}[3]{#1\ \Vdash\ {#2} \ \cseq\ #3}

\newcommand{\gclose}[2]{#1 \mathrel{\rightarrow}#2}
\newcommand{\gcom}[3]{#1 \mathrel{\rightarrow}#2(#3)}
\newcommand{\gsel}[4]{#1 \mathrel{\rightarrow}#2.\m{case}(#3,#4)}
\newcommand{\gbang}[3]{{!}\gcom{#1}{#2}{#3}}
\newcommand{\gexistsforall}[4]{#1 \mathrel{\rightarrow}#2.(#3)#4}
\newcommand{\gabsurd}{\topzero}

\newcommand{\comsel}[5]{#1 \mathrel{\rightarrow} #2.#3(#4; #5)}
\newcommand{\lleft}{\m{inl}}
\newcommand{\lright}{\m{inr}}

\newcommand{\topzero}[2]{\emptyCase{\gfromto{#1}{#2}}}
\newcommand{\gto}{\to}
\newcommand{\gfromto}[2]{#1 \gto #2}
\newcommand{\Casebase}[2]{{{#1}.\m{case}(#2)}}
\newcommand{\emptyCase}[1]{{\Casebase{#1}{}}}
\newcommand{\axiom}[4]{{#1}^{#2} \gto {#3}}

\newcommand{\bang}{\boldsymbol{!}}
\newcommand{\quest}{\boldsymbol{?}}

\newcommand{\dual}[1]{#1^\perp}
\newcommand{\globalaxiom}[4]{{#1}^{#2} \to {#3}}
\newcommand{\globalaxiomp}[4]{\globalaxiom{#1}{#2}{#3}{#4}}
\newcommand{\phl}{\hl}

\newcommand{\reducesto}{\longrightarrow}

\newcommand{\projextr}{\rightleftharpoons}

\newcounter{ncomm}

\newcommand{\queryt}[1]{\mathord{?}}

\newcommand{\mypar}[1]{\noindent{\bf #1.}}

\begin{document}
\pagestyle{plain}

\title{Multiparty Classical Choreographies}

\author{Marco Carbone\inst{1}
  \and Lu\'\i s Cruz-Filipe\inst{2}
  \and Fabrizio Montesi\inst{2}
  \and Agata Murawska\inst{1}
}
\institute{
  \begin{tabular}{cc}
    $^1$IT University of Copenhagen
    & 
    \qquad
    $^2$University of Southern Denmark
    \\
    \email{\{carbonem,agmu\}@itu.dk}
    &
    \qquad
    \email{\{lcf,fmontesi\}@imada.sdu.dk}
    \\[-2.9mm]
  \end{tabular}
}

\maketitle

\begin{abstract}
We present Multiparty Classical Choreographies (MCC), a language model
where global descriptions of communicating systems (choreographies)
implement typed multiparty sessions. Typing is achieved by
generalising classical linear logic to judgements that explicitly
record parallelism by means of hypersequents. Our approach unifies
different lines of work on choreographies and processes with
multiparty sessions, as well as their connection to linear
logic. Thus, results developed in one context are carried over to the
others. Key novelties of MCC include support for server invocation in
choreographies, as well as logic-driven compilation of choreographies
with replicated processes.
%
%

\end{abstract}

\section{Introduction}\label{sec:intro}
Choreographic Programming \cite{M13:phd} is a programming paradigm
where programs, called {\em choreographies}, define the intended
communication behaviour of a system based on message passing, using an
``Alice and Bob'' notation, rather than the behaviour of each
endpoint.  Choreographies are useful for several reasons: they give a
succinct description, or {\em blueprint}, of the intended behaviour of
a whole system, making the implementation less error-prone. Then,
correct-by-construction distributed implementations can be synthesised
automatically by means of {\em projection}, a compilation algorithm
that generates the code for each endpoint described in the
choreography \cite{CHY12,CM13}. Reversely, it is often possible to
obtain a choreography from an endpoint implementation by means of {\em
  extraction}, providing a precise blueprint of a distributed system.

Choreographic programming has a deep relationship with the proof
theory of linear logic \cite{CMS18}. Specifically, choreographic
programs can be seen as terms describing the reduction steps of cut
elimination in linear logic (choreographies as cut reductions).  The
key advantage of this result is that it provides a logical
reconstruction of two useful translations, one from choreographies to
processes ({\em projection}, or synthesis) and another from processes
to choreographies (\emph{extraction}) -- this is obtained by
exploiting 
the correspondence between intuitionistic linear logic and a variant
of the $\pi$-calculus \cite{CP10}.  These translations can be used to
keep process implementations aligned with the desired communication
flows given as choreographies, whenever code changes are applied to
any of the two.  This kind of alignment is a desirable property in
practice, e.g., it is the basis of the Testable Architecture
development lifecycle for web services \cite{savara:website}.

Unfortunately, the logical reconstruction of choreographies in
\cite{CMS18} covers only the multiplicative-additive fragment of
intuitionistic linear logic, limiting its practical applicability to
simple scenarios. The aim of this paper is to push the boundaries of
this approach 
towards more realistic scenarios with sophisticated
features. In this article, we define a model, strictly related to
classical linear logic, that allows for {\em replicated
  services}, 
and  {\em multiparty sessions}.

Reaching our aim is challenging for both design and technical
reasons. 
In the multiplicative-additive fragment of linear logic considered in
\cite{CMS18}, all reductions intuitively match choreographic terms
explored in previous works on choreographies, i.e., communication of a
channel and branch selection \cite{CM13}. This is not the case for the
exponential fragment, which yields {\em reductions never considered
  before in choreographies}, e.g., explicit garbage collection of
services and server cloning (see \textsf{kill} and \textsf{clone}
operations).
To bridge this gap, we exploit the fact that these operations occur
naturally in the process language and, through the logic, can be
reflected to choreographic primitives for management of services as
explicit resources that can be duplicated, used, or destroyed.
%
We show that the reductions for these terms correspond to the
principal cut reductions for exponentials in classical linear logic.
Typing guarantees that resource management is safe, e.g., no destroyed
resource is ever used again.  

In \cite{CMS18}, all sessions (protocols) have exactly two
participants. This works well in intuitionistic linear logic, where
sequents are two-sided: two processes can be connected if one
``provides'' a behaviour and the other ``needs'' it. This is verified
by checking identity of types, respectively between a type on the
right-hand side of the sequent of the first process and a type of the
left-hand side of the sequent for the second.  To date, it is still
unclear how identity for two-sided sequents can be generalised to
multiparty sessions, where a session can have multiple participants
and thus we need to check compatibility of multiple types.  Instead,
this topic has been investigated in the setting of classical linear
logic, where multiparty compatibility is captured by coherence, a
generalisation of duality \cite{CMSY17}.  Therefore, our formulation
of Multiparty Classical Choreographies (MCC) is based on classical
linear logic.  In order to bridge choreographies to multiparty
sessions, we introduce a new session environment, which records the
types of multiparty communications performed by a choreography as
\emph{global types} \cite{HYC16}.  The manipulation of the session
environment reveals that typing a choreography with multiparty
sessions corresponds to \emph{building the coherence proofs for typing
  its sessions}.  Since a proof of coherence is the type compatibility
check required by the multiparty version of cut in classical linear
logic, our result generalises the choreographies as cut reductions
approach to the multiparty case as one would expect, providing further
evidence of the robustness of this idea.
The final result of our efforts is an expressive calculus for
programming choreographies with multiparty sessions and services,
which supports both projection and extraction operations \emph{for all
  typable programs}.

\section{Preview}\label{sec:preview}
We start by introducing MCC informally, focusing on modelling a
protocol inspired by OpenID~\cite{openid}, where a client
authenticates through a third-party identity provider. MCC offers a
way of specifying protocols in terms of global types. For example, our
variant of 
OpenID 
can be specified by the global type $G$: 
\begin{displaymath}\small
  \begin{array}{llll}
    \gcom u {rp}{\mathsf{String}}; 
    \gcom u {ip}{\mathsf{String}}; 
    \gcom u {ip}{\mathsf{PWD}}; 
    \gsel {ip}{rp}{\;\gcom u{rp}{\mathsf{String}};G_1}{\ G_2}
  \end{array}
\end{displaymath}
This protocol concerns three endpoints (often called roles in
literature) denoted by $u$ (user), $rp$ (relaying party) and $ip$
(identity provider). The user starts by sending its login string to
both $rp$ and $ip$. Then, it sends its password to $ip$ which will
either confirm or reject $u$'s authentication to $rp$. If the
authentication is successful then the user will send an evaluation of
the authentication service to $rp$, and then complete as the
unspecified protocol $G_1$. Otherwise, if the password is wrong, then
the protocol continues as $G_2$.  The specification given by the
global type $G$ can be used by a programmer during an implementation.
In MCC, we could give an implementation in terms of the choreography:
\begin{displaymath}\small
  \begin{array}{l@{\qquad\qquad} ll}
    \gclient u{rp,ip}{} & \texttt{// $u$ starts protocol with rp and ip}\\
    \com{u}{user_u}{rp}{user_{rp}}{
    } & \texttt{// $u$ sends its login to $rp$}\\ 
    \com{u}{login_{u}}{ip}{login_{ip}}{} & \texttt{// $u$ sends its login to $ip$}\\ 
    u(pwd_{u})\ \rightarrow\ ip(pwd_{ip}); & \texttt{// $u$ authenticates with $ip$}\\[1mm]
    \multicolumn{2}{l}{
    ip\ \rightarrow rp.
    \left\{
    \begin{array}{lll}\!\!
      \begin{array}{ll@{\qquad}ll}
        \textsf{inl}:
        & \gclient {u'}s{} & \texttt{// $u'$ starts protocol with $s$}\\
        & \com{u'}{rep_{u'}}{s}{rep_{s}}{}& \texttt{// $u'$ sends report to $s$}\\
        & \com{s}{ack_s}{u'}{ack_{u'}}{}& \texttt{// $s$ acknowledges to $u'$}\\
        & \com{u}{rep_{u}}{rp}{rep_{rp}}{P},  & \texttt{// $u$ sends report to $rp$}\\
        \textsf{inr}: & Q & \texttt{// authentication fails}
      \end{array}
    \end{array}
    \right\}
    }
  \end{array}
\end{displaymath}
Each line is commented with an explanation of the performed action. We
observe that two different protocols are started. The first line
starts the OpenID protocol between $u$, $rp$ and $ip$ described
above. Moreover, after branching, the choreography starts another
session between the user (named $u'$) and a server $s$ that is used
for reviewing the authentication service given by $ip$. In this case,
the protocol used is
$G'=\gcom{u'}{s}{\mathsf{String}}; \gcom s{u'}{\mathsf{String}}; G_3$,
for some unspecified $G_3$.  We leave undefined the case in which the
identity provider receives a wrong password (term $Q$).

In this work, we show how a choreography that follows a protocol such
as $G$ can be expressed as a proof in a proof theory strictly related
to classical linear logic. Moreover, thanks to proof transformations,
the choreography above can be projected into a parallel composition of
endpoint processes, each running a different endpoint. As an example,
the endpoint process for the user would correspond to the process
$P_u$, defined as
\begin{displaymath}\small
  \begin{array}{l}
    \client {u}{rp,ip}{}\co{u} (user_u); \co{u} (login_u);
    \co{u} (pwd_u);
    \client {u'}{s}{}\co{u'} (rep_{u'}); u'(ack_{u'});
    \co{u} (rep_{u}); R
  \end{array}
\end{displaymath}
which mimics the behaviour of $u$ and $u'$ specified in the
choreography. Operator $\m{use}$ is used to start a session, while the
other two operators utilised above are for in-session
communication. Similarly, we can have the endpoint processes for $rp$,
$ip$ and $s$: 
\begin{displaymath}\small
  \begin{array}{c}
    P_{rp}=\m{srv}\, {rp}; rp (user_{rp}); rp.\mathsf{case}(rp (rep_{rp}); R_1,Q_1)
    \\
    P_{ip}=\m{srv}\, {ip}; ip (login_{ip});
    ip (pwd_{ip}); R_2
    \qquad
    P_{s}= \m{srv}\, {s}; s(rep_{s}); \co{s}(ack_{s}); R_3
  \end{array}
\end{displaymath}

\section{GCP with Hypersequents}\label{sec:gcp}
%
In this section, we present the {\em action fragment} of MCC, where we
only consider local actions, e.g., inputs or outputs. The action
fragment is a variant of Globally-governed Classical Processes
(GCP)~\cite{CLMSW16} whose typing rules use
hypersequents. 
In the remainder, we denote a vector of endpoints $x_1,\ldots,x_n$ as
$\tilde x$ or $(x_i)_i$.
\paragraph{Syntax.} The action fragment is a generalisation of
Classical Processes~\cite{W14} that supports multiparty session
types. As hinted in \S\ref{sec:preview}, when writing a program in our
language, we do not identify sessions via channel names, but rather we
name sessions' \emph{endpoints}. Each process owns a single endpoint
of a session it participates in. The complete syntax is given by the
following grammar:
\begin{displaymath}\scriptsize
  \begin{array}{rll@{\qquad\qquad}llllllllll}
    P ::= & \rule{1.4mm}{0mm}
    \hl{\axiom{x}{A}{y}{B}} & \text{link} 
    & \mid \hl{\genclient} & \text{client}
    \\[1mm]
    & \mid \hl{ P\pp Q}     & \text{parallel}
    & \mid \hl{\gensrv}& \text{server}
    \\[1mm]
    & \mid \hl{\res {{\til x}:G} P} & \text{restriction}
    & \mid \hl{\genshutdown} & \text{server kill}
    \\[1mm]
    & \mid \hl{\gensend} & \text{send}
    & \mid \hl{\genclone}    & \text{server clone}
    \\[1mm]
    & \mid \hl{\genrecv} & \text{receive}
    \\[1mm]
    & \mid \hl{\genclose} & \text{close session}
    \\[1mm]
    & \mid \hl{\genwait} & \text{receive close}
    & \mid \hl{\gencase} & \text{branching}
    \\[1mm]
    & \mid \hl{\geninl} & \text{left selection}
    & \mid \hl{\geninlinr} & \text{general selection}
    \\[1mm]
    & \mid \hl{\geninr} & \text{right selection}
    & \mid \hl{\gennocase} & \text{empty choice}
  \end{array}
\end{displaymath}
With a few exceptions, the terms above are identical to those of
GCP. {\em For space restriction reasons, we only discuss the key
  differences.} Parallel and restriction constructs form a single term
$(\nu \tilde x:G)(P\pp Q)$ in the original GCP. The link process
$\axiom{x}{A}{y}{B}$ is a forwarder from $x$ to $y$. We further allow
the general selection $\geninlinr$, denoting a process that
non-deterministically selects a left or a right branch.
For services, an endpoint $x$ may kill all servers by executing the
action $\genshutdown$, or duplicate them by means of $\genclone$ --
these operations were silent in the original 
GCP.
In cloning, the new server copies are replicated at fresh endpoints,
ready to engage in a session with new endpoint $x'$. More generally,
we follow the convention of~\cite{W14}, denoting the result of
refreshing names in $Q$ by $Q'$ (changing each $x\in\m{fv}(Q)$ into a
fresh $x'$).

\paragraph{Types.} Types, used to ensure proper behaviour of
endpoints, are defined as:
\begin{displaymath}\scriptsize
  \begin{array}{rllll@{\qquad}lllllll}
    A ::= & \rule{1.5mm}{0mm}\hl{A\tensor B} & \text{output} 
    & \mid \hl{A\parr B} & \text{input} 
    & G ::= & \rule{1.5mm}{0mm}\hl{\gcom {\til x}{y}{G};H} & (\tensor\parr)\\[1mm]
          & \mid \hl{A\oplus B} & \text{selection} 
    & \mid \hl{A\with B} & \text{choice}
    && \mid \hl{\gsel x{\til y}{G}{H}} & (\oplus\with) \\[1mm]
          & \mid \hl{!A} & \text{server} 
    & \mid \hl{?A}       &\text{client} 
    && \mid \hl{\gbang{x}{\til y}{G}}  & (!?) \\[1mm]
          & \mid \hl{\one} & \text{close} 
    & \mid \hl{\bot}       & \text{wait} 
    && \mid \hl{\gclose{\til x}{y}} & (\one\bot) \\[1mm]
          & \mid \hl{0} & \text{false} 
    & \mid \hl{\top}    & \text{empty} 
    && \mid \hl{\gabsurd{x}{\til y}} & (0\top) \\[1mm]
          & \mid \hl{X} & \text{variable} 
    & \mid \hl{X^\bot} & \text{dual variable} 
    && \mid \hl{\globalaxiom{x}{A}{y}{B}} & (\textsc{axiom})
  \end{array}
\end{displaymath}
In the multiparty setting, types can be split into \emph{local types}
$A$, which specify behaviours of a single process, and \emph{global
  types} $G$, which describe interaction within sessions (and
choreography actions). Again, most global types correspond to pairs of
local types, the exception being the global axiom type, describing a
linking session (restricted by typing to type variables and their
duals).  Local type operators are based on connectives from classical
linear logic -- thus, ${A\tensor B}$ is the type of a process that
outputs an endpoint of type $A$ and continues with type $B$, whereas
${A\parr B}$ is the type of a process that receives endpoints of type
$A$ and is itself ready to continue as $B$. The corresponding global
type ${\gcom {\til x}{y}{G};H}$ types the interaction where each of
the processes owning an endpoint $x_i$ sends their new endpoint to
$y$. 
Type $0$ is justified by the necessity of having a type dual to
$\top$, while the rule $0\top$ is essential for the definition of
coherence.  Type variables are used to represent concrete datatypes.
It is worth noting that the logic formulas in our type system enjoy
the usual notion of duality, where a formula's dual is obtained by
recursively replacing each connective by the other one in the same row
in the table above. For example, the dual of $!(A\tensor0)$ is
$?(A^\bot\parr\top)$, where $A^\bot$ is the dual of formula $A$.

\paragraph{Typing.}
We type our terms in judgements of the form $\judge\Sigma P\Psi$,
where:
(i) $\Sigma$ is a set of session typings of the form
$\pair{(x_i)_i}{G}$;
(ii) $P$ is a process;
and,
(iii) $\Psi$ is a hypersequent, a set of classical linear logic
sequents.
Intuitively, ${\judge\Sigma P\Psi}$ reads as {\em ``$\Psi$ types $P$
  under the session protocols described in $\Sigma$.''}

\noindent Given a judgment
$\judge \Sigma P {\Psi \pp \vdash \Gamma, x : A}$, checking whether
$x$ is \emph{available} -- not engaged in a session -- is implicitly
done by verifying that $x$ does not occur in the domain of
$\Sigma$. Note that names cannot occur more than once in $\Sigma$:
each endpoint $x$ may only belong to (at most) one session
$G$. Hypersequents $\Psi_1, \Psi_2$ and sets of sessions
$\Sigma_1, \Sigma_2$ can only be joined if their domains do not
intersect. Moreover, we use indexing in different ways:
$\big(\judge{\Sigma_i}{P_i}{\Psi_i}\big)_i$ denotes several judgements
$\judge{\Sigma_1}{P_1}{\Psi_1}$, \ldots,
$\judge{\Sigma_n}{P_n}{\Psi_n}$; indexed pairs $(\pair {x_i}{A_i})_i$
are a set of pairs $\pair {x_1}{A_1}$, \ldots, $\pair {x_n}{A_n}$;
and, finally, $(\vdash \Gamma_i)_i$ denotes the hypersequent
$\vdash \Gamma_1\ |\ \ldots\ |\ \vdash\Gamma_n$.

In order to separate restriction and parallel (reasons for this
separation will be explained in \S~\ref{sec:choreo}), we split the
classical linear logic $\m{Cut}$ rule into
two: 
{\scriptsize
\begin{align*}
 &\infer[\m{Conn}]
      {
        \judge{(\Sigma_i)_i, \pair{(x_i)_i}{G}}{(P_i)_i}{\left(\Psi_i \pp \seq \Gamma_i, \pair{x_i}{A_i} 
          \right)_i}
      }
      {
        \left(
        \judge{\Sigma_i}{P_i}{\Psi_i \pp \seq \Gamma_i, \pair{x_i}{A_i}}
        \right)_i
        \!\!&\!\!
        G \gseq (\pair{x_i}{A_i})_i
      }
  &\,\, \!\!\!\!
    &\infer[\m{Scope}]
      {\judge{\Sigma}{\res {{\til x}:G} P}{\Psi \pp \seq \left( \Gamma_i \right)_i}}
      {\judge{\Sigma, \pair{(x_i)_i}{G}}{P}{\Psi \pp \left(\seq \Gamma_i, \pair{x_i}{A_i} \right)_i}}
\end{align*}        
}%
Rule $\m{Conn}$ is used for merging proofs that provide {\em coherent}
types (we address coherence below), but without removing them from the
environment. 
Since such types need to remain in the conclusion of the rule, we need
to use hypersequents. 
%
The sequents involved in a session get merged once a $\m{Scope}$ rule
is applied. This hypersequent presentation is similar to a classical
linear logic variant of \cite{CMS18} with sessions explicitly
remembered in a separate context $\Sigma$. 

%
\begin{figure*}[t]
  \begin{eqnarray*}
    &\infer[\tensor\parr]
      {
      \hl{\gcom {\til x}{y}{G};H}
      \ \ \gseq\ \
      \Gamma, \
      (\pair{ {x_i}}{ {A_i \tensort{y} B_i}})_i,\ 
      \pair{ y}{{C \parrt{\til x} D}}
      }
      {
      \hl G 
      \ \gseq\
      ( \pair {{x_i}}{{A_i}} )_i,\ 
      \pair{y}{C}
      &
        \hl {H}
        \ \gseq\ \Gamma,\ 
        ( \pair {{x_i}}{{B_i}} )_i,\ 
        \pair{y}{D}
        }
        \quad
        \infer[\one\bot]
        {
        \hl{\gclose{\til x}{y}}
        \gseq
        ( \pair {{x_i}}{{\onet y}} )_i,
        \pair {y} {{\bott{\til x}}}
        }
        {
        }
      \\[1ex]
      &\infer[\oplus\with]
      {
      \hl{\gsel x{\til y}{G_1}{G_2}}
      \gseq \Gamma, \pair{x}{{A\oplust{\til y} B}},
      (\pair{\hl {y_i}}{{C_i \witht x D_i}})_i
      }
      {
      \hl{G_1} \gseq \Gamma, \pair {x}{A}, (\pair {{y_i}}{{C_i}})_i
      &
        \hl{G_2} \gseq\Gamma, \pair {x}{B}, (\pair {\hl{y_i}}{{D_i}})_i
        }
        \quad
        \infer[!?]
        {
        \hl{\gbang{x}{\til y}{G}}
        \gseq
        \pair {\hl x}{{\queryt{\til y}A}},
        (\pair{\hl{y_i}}{{\bangt x B_i}})_i
        }
        {
        \hl G
        \gseq
        \pair {\hl x}{A},
        (\pair{\hl{y_i}}{{B_i}})_i
        }
      \\[1ex]
      &\infer[0\top]
      {
      \hl{\gabsurd{x}{\til y}} \gseq
      \hl\Gamma, \pair {x}{{\zerot{\til y}}}, ( \pair {{y_i}}{{\topt{x}}} )_i
      }
      {}
      \qquad
          \infer[\textsc{Axiom}]
      {\hl{\globalaxiom{x}{A}{y}{A^\bot}} \gseq x: {A}, y:{A^\bot}}
      {
      A^\top = X \mbox{ or } A = X^\bot 
      }
  \end{eqnarray*}
  \caption{Coherence rules.}
  \label{fig:coherence}
\end{figure*}
%

%
{\em Coherence} is a generalisation of duality~\cite{CLMSW16} to more
than two parties: when describing a multiparty session, simple duality
of types does not suffice to talk about their compatibility. 
In Fig.~\ref{fig:coherence}, we report the rules defining the
coherence relation $\gseq$. We do not describe these here in detail,
as they remain unchanged compared to the original GCP presentation,
with the exception of the axiom rule which is only applicable to
atomic types in our system.

%
%
\begin{figure}[t]
  \[
\begin{array}{c}
\infer[\m{Ax}]
    {\judge{\cdot}{\axiom{x}{A}{y}{A^\bot}}{\seq \pair xA, \pair y{A^\bot}}}
    { A = X \mbox{ or } A = X^\bot } \quad
  \infer[\tensor]
    {\judge{\Sigma_1,\Sigma_2}{\gensend}
            {\Psi_1\pp \Psi_2\pp \seq {\Gamma_1,\Gamma_2,\pair x {A\tensort y B}}}}
    { \judge{\Sigma_1} P {\Psi_1\pp\seq {\Gamma_1, \pair{x'}A}}
      &
      \judge{\Sigma_2} Q {\Psi_2\pp\seq {\Gamma_2,\pair xB}}
    }\\[2mm]
    \infer[\parr]
          {\judge{\Sigma}{\genrecv} { \Psi\pp \seq {\Gamma, \pair y{A\parrt{\til x} B}} } }
          {\judge{\Sigma}P {\Psi\pp \seq {\Gamma, \pair{y'}A, \pair yB} }}
    \quad
    \infer[\bot]
          {\judge{\Sigma}{\genwait}
                  {\Psi\pp \seq \Gamma, \pair y {\bott {\til x}}}}
          {\judge \Sigma {P} { \Psi\pp \seq \Gamma } }
    \\[2mm]
    \infer[\one]
           {\judge{\cdot}{\genclose} {\seq \pair x {\onet y}}}
           {}
    \quad
    (\text{no rule for } \nil)
    \qquad
    \infer[\top]
          {
            \judge \cdot {\gennocase} {\seq \Gamma, \pair x{\top}}
          }
          {\m{vars}(\Gamma) = \til u}\\[2mm]
    \infer[\oplus_1]
        {
          \judge{\Sigma}{\geninl}
                {\Psi \pp \seq \Gamma, \pair {x}{A\oplust{\til y} B}}
        }
        {
          \judge \Sigma P 
                 {\Psi \pp \seq \Gamma, \pair {x}{A}}
        }
    \qquad
    \infer[\oplus_2]
        {
          \judge\Sigma{\geninr} 
                {\Psi \pp \seq \Gamma, \pair {x}{A\oplust{\til y} B}}
        }
        {
          \judge\Sigma Q
                 {\Psi \pp \seq \Gamma, \pair {x}{B}}
        }
    \\[2mm]
    \infer[\oplus]
        {\judge{\Sigma}{\inlinr{x}{P}{Q}}{\Psi \pp \seq \Gamma, \pair {x}{A\oplus B}}}
        {\judge{\Sigma}{P}{\Psi \pp \seq \Gamma, \pair x A} &
          \judge{\Sigma}{Q}{\Psi \pp \seq \Gamma, \pair x B}}
        \quad
        \infer[!]
      {
        \judge \cdot {\gensrv}
               { \seq {?\Gamma}, \pair {y}{\bang A} }
      }
      {
        \judge \cdot {P}
               { \seq {?\Gamma}, \pair {y}{A} }
      }\\[2mm]
    \infer[\with]
      {
        \judge{\Sigma}{\gencase} 
              {\Psi \pp \seq {\Gamma, \pair {y}{A\witht{x} B}}}
      }
      {
        \judge\Sigma P {\Psi \pp  \seq \Gamma, \pair {y}{A}}
        &
        \judge\Sigma Q {\Psi \pp \seq \Gamma, \pair {y}{B}}
      }    \quad
    \infer[?]
          {
            \judge \Sigma {\genclient}
                   { \Psi\pp \seq \Gamma, \pair {x}{\quest A} }
          }
          {
            \judge \Sigma {P}
                   { \Psi\pp \seq \Gamma, \pair {x}{A} }
          }
    \\[2mm]
    \infer[\m{Weaken}]
      {
        \judge \Sigma {\genshutdown}
               { \Psi\pp \seq \Gamma, \pair {x}{{\quest A}} }
      }
      {
        \judge \Sigma {P} { \Psi\pp \seq \Gamma }
      }
    \qquad
    \infer[\m{Contract}]
          {
            \judge \Sigma {\genclone}
                   { \Psi\pp \seq \Gamma, \pair {x}{{\quest A}}}
          }
          {
            \judge \Sigma P
                   { \Psi \pp \seq \Gamma, \pair {x}{{\quest A}}, \pair {x'}{{\quest A}}}
          }
\end{array}\]
\caption{Rules for the action fragment.}
\label{fig:mcc_action}
\end{figure}
The remaining typing rules for the action fragment, presented in
Fig.~\ref{fig:mcc_action} are identical to those of GCP with the
exception that a context in GCP may be distributed among several
sequents here. For example, rule $\tensor$ takes two sequents
$\vdash\Gamma_1, \pair{x'}A$ and $\vdash\Gamma_2, \pair{x}B$ from two
different hypersequents, and merges them into
$\vdash\Gamma_1,\Gamma_2, \pair{x}A\tensor B$, as in classical linear
logic. However, elements of $\Gamma_1$ and $\Gamma_2$ may be connected
through $\Sigma_1$ and $\Sigma_2$ to other parts of $\Psi_1$ and
$\Psi_2$ respectively (as a result of previously applied \m{Conn}).
Note that the rules of this fragment work only with processes not
engaged in any session, since the endpoints explicitly mentioned in
proof terms cannot occur in the domain of $\Sigma$:
this is an implicit check in all rules of Fig.~\ref{fig:mcc_action}.
Rule $\top$ introduces a single sequent $\Gamma,x:\top$, allowing for
any $\Gamma$. The proof term $\gennocase$ keeps track of the endpoints
introduced in $\Gamma$: it ensures that all endpoints in the typing
are mentioned in the proof term, which is useful when defining
semantics. In this article, we restrict the axiom to only type
variables (see \S\ref{sec:related}).
%

\paragraph{Semantics.}
\begin{figure}
  \[\begin{array}{l  @{\qquad}l} 
      \begin{array}{ll}
        \hl{(\til  P \pp Q) \pp \til S}
        \equiv \hl{\til P \pp (Q \pp \til S)}&\\[1mm]
        \hl {(\gensend) \pp \til S}
        \equiv \hl {\send{x}{x'}{y',\til z'}{(P\pp \til S)}{Q}} &\\[1mm]
        \hl {(\gensend) \pp \til S}
        \equiv \hl {\send{x}{x'}{y',\til z'}{P}{(Q \pp \til S)}} & \\[1mm]
        \hl{\genrecv \pp \til Q}
        \equiv \hl{\recv {y}{\til x'}{y'}{(P \pp \til Q)}} &\\[1mm]
        \hl{\genwait \pp \til Q}
        \equiv \hl{\wait{y}{(P \pp \til Q)}} & \\[1mm]
        \hl{\geninl \pp \til Q}
        \equiv \hl{\inl{x}{(P \pp \til Q)} } & \\[1mm]
        \hl{\geninr \pp \til Q}
        \equiv \hl{\inr{x}{(P \pp \til Q)} } & \\[1mm]
        \hl{\geninlinr \pp \til S} 
        \equiv\\
        \qquad \qquad\qquad\hl{\inlinr{x}{P\pp\til S}{Q\pp\til S}}\\[1mm]
        \hl{\gencase \pp \til S}
        \equiv \hl{\Case {y}{P \pp \til S}{Q \pp \til S}} & \\[1mm]
        \hl{\genclient \pp \til Q}
        \equiv \hl{\client x{\til y}{(P \pp \til Q)}} & \\[1mm]
        \hl{\genshutdown \pp \til Q}
        \equiv \hl{\shutdown x{\til y}{(P \pp \til Q)}} & \\[1mm]
        \hl{\genclone \pp \til Q}
        \equiv \hl{\clone x{\til y}{x'}{\til y'}{(P \pp \til Q)}} 
      \end{array}
      \ \ 
      \begin{array}{ll}
        \hl{\res {{\til x}:\!G} (P \pp \til Q)} \equiv \hl{\res {{\til
        x}:G} P \pp \til Q} &\\[1mm]
        \hl{\res {{\til x}:\!G} \res {{\til y}:\!H} P} \equiv \hl{\res
        {{\til y}:\!H} \res {{\til x}:\!G} P} &\\[1mm]
        \hl {\res {{\til w}:\! G} (\gensend)} \equiv\\[1mm]
        \qquad\qquad \hl
        {\send{x}{x'}{y',\til z'}{ \res{{\til w}:\! G} P}{Q}}
        \quad \big(\exists i . w_i \in \m{fv}(P)\big) \\[1mm]
        \hl {\res {{\til w}:\! G} (\gensend)} \equiv \\[1mm]
        \qquad\qquad \hl{\send{x}{x'}{y',\til z'}{P}{\res{{\til w}:\! G} Q}}
        \quad \big(\exists i. w_i \in \m{fv}(Q)\big)\\[1mm]
        \hl {\res {{\til w}:\!G} (\genrecv)} \equiv \hl {\recv
        {y}{\til x'}{y'}{\res {{\til w}:\!G} P}} \\[1mm]
        \hl {\res {{\til w}:\!G} (\geninl)} \equiv \hl {\inl{x}{\res {{\til
        w}:\!G} P}} \\[1mm]
        \hl {\res {{\til w}:\!G} (\geninr)} \equiv \hl {\inr{x}{\res
        {{\til w}:\!G} Q}} \\[1mm]
        \hl{\res{{\til w}:\!G} (\geninlinr)}
        \equiv\\
        \qquad\qquad\hl{\inlinr{x}{\res{{\til w}:\!G} P}{\res{{\til w}:\!G}
        Q}}\\[1mm]
        \hl {\res {{\til w}:\!G} (\gencase)} \equiv \\
        \qquad\qquad\hl {\Case
        {y}{\res {{\til w}:\!G}P}{\res {{\til w}:\!G}Q}} & \\[1mm]
   \hl {\res {{\til w}:\!G}(\genclient)} \equiv \hl {\client x {\til y} {\res {{\til w}:\!G}P}} & \\[1mm]
   \hl {\res {{\til w}:\!G} (\genshutdown)} \equiv \hl {\shutdown x {\til y} {\res {{\til w}:\!G}P}} & \\[1mm]
   \hl {\res {{\til w}:\!G} (\genclone)} \equiv \hl {\clone x {\til y}{x'}{\til y'} {\res {{\til w}:\!G}P}} &
      \end{array}\\
      \multicolumn{2}{l}{\qquad\qquad\qquad\qquad\qquad\hl{\res{{\til
      x}:\!G}(\gensrv \pp \til Q)} \equiv \hl{ \srv{y}{}{\res{{\til
      x}:\!G}(P \pp \til Q)}}}\\[1mm]
      \multicolumn{2}{l}{\qquad\qquad\qquad\qquad\qquad\hl{\res{{\til
      z z}:\!G}{(\nocase{x}{}{\til u, z}\pp \til Q)}} \equiv
      \hl{\nocase{x}{}{\til u, \til v}} \qquad \qquad \text{where
      }\til v = \m{vars}(\til Q) \setminus \til z}
    \end{array}\]
  \caption{Equivalences for commuting the action fragment with \m{Conn} and \m{Scope}. All rules
    assume that both sides of the equation are typable in the same context.}
  \label{fig:commconv_action}
\end{figure}

The semantics of the action fragment is almost identical to that of
standard GCP.  It is obtained from cases of the proof of cut
elimination: the principal cases 
describe reductions ($\reducesto$), while the permutations of rule
applications give rise to the rules for structural equivalence
($\equiv$), reported in Fig.~\ref{fig:commconv_action}. Note that as
we are interested only in commuting conversions of typable programs,
there are certain cases where the correct equivalence can be found
only by looking at the typing derivation which contains information
that is not part of the process term.  Under $\equiv$, parallel
distributes safely over $\m{case}$ (because only the actions of one
branch are going to be executed). A similar mechanism can be found in
the original presentation of Classical Processes \cite{W14}, and was
later demonstrated to correspond to a bisimulation law in \cite{A17}.
%
\begin{figure*}[!bt] 
  \[
    \begin{array}{lrl}
      \multicolumn{3}{l}{
      \phl{\gcpres{\til x, y, \til z}
      {\gcom{\til x}{y}{G};H}
      {\big((\send{x_i}{x_i'}{y',\til x'_{\setminus i}}{P_i}{Q_i})_i
      \pp \recv{y}{\til x'}{y'}R \pp \til S\big)}}}
      \\
      & \reducesto {}
      \\
      \multicolumn{3}{r}{
      \phl{\gcpres{\til x', y'}
      {G\{\til x'/\til x,y'/y\}} {\big(\til P \pp \gcpres{\til x, y, \til z}{H}
      {(\til Q \pp R \pp \til S)}\big)}}
      }
    \\
    \phl{\gcpres{\til x, y}{\gclose{\til x}{y}} {\big((\close{x_i})_i \pp \genwait}\big)} & \reducesto & \phl P
    \\
    \multicolumn{3}{l}{
      \phl{\gcpres{x, \til y, \til z}{\gsel{x}{\til y}{G}{H}} \big(\geninl
            \pp (\Case{y_i}{Q_i}{R_i})_i \pp \til S}\big)}
    \\
    & \reducesto & \phl{\gcpres{x, \til y, \til z}{G}{\big(P \pp \til Q \pp \til S\big)}}
    \\
    \multicolumn{3}{l}{
      \phl{\gcpres{x, \til y, \til z}{\gsel{x}{\til y}{G}{H}} \big(\inr{x}{P}
            \pp (\Case{y_i}{Q_i}{R_i})_i \pp \til S\big)}}
    \\
    & \reducesto & \phl{\gcpres{x, \til y, \til z}{H}{\big(P \pp \til R \pp \til S\big)}}
    \\
    \multicolumn{3}{l}{
      \phl{\gcpres{x, \til y, \til z}{\gsel{x}{\til y}{G}{H}} \big(\geninlinr
            \pp (\Case{y_i}{R_i}{S_i})_i \pp \til T\big)}}
    \\
    & \reducesto & \phl{\gcpres{x, \til y, \til z}{G}{\big(P \pp \til R \pp \til T\big)}}
    \\
    \multicolumn{3}{l}{
      \phl{\gcpres{x, \til y, \til z}{\gsel{x}{\til y}{G}{H}} \big(\geninlinr
            \pp (\Case{y_i}{R_i}{S_i})_i \pp \til T\big)}}
    \\
    & \reducesto & \phl{\gcpres{x, \til y, \til z}{H}{\big(Q \pp \til S \pp \til T\big)}}
    \\
    \phl{\gcpres{x, \til y}{\gbang{x}{\til y}{G}} \big(\genclient \pp (\srv{y_i}{x}{Q_i})_i\big)}
    & \reducesto & \phl{\gcpres{x,\til y}G {\big(P \pp \til Q\big)}}
    \\
    \phl{\gcpres{x, \til y}{\gbang{x}{\til y}{G}} \big(\genshutdown \pp (\srv{y_i}{x}{Q_i})_i\big)}
    & \reducesto & \phl{\left(\srvkill{u_j}\right)_j \pp P}\\
    \multicolumn{3}{r}{\hfill \mbox{where }\forall i. \forall v_i \in \m{fv}(Q_i) . v_i \neq y_i \Rightarrow \exists j . v_i = u_j}
    \\
    \phl{\gcpres{x, \til y}{\gbang{x}{\til y}{G}} \big(\genclone \pp (\srv{y_i}{x}{Q_i})_i\big)}
    & \reducesto & \\
    \multicolumn{3}{r}{
      \phl{\left(\cloneo{u_j}{u_j'}\right)_j; \gcpres{x, \til y}{\gbang{x}{\til y}{G}}
        {\gcpres{x', \til y'}{\gbang{x'}{\til y'}{G\{x'/x,\til y'/\til y\}}}
          {\big(P \pp (\srv{y_i}{x}{Q_i})_i} \pp (\srv{y'_i}{x'}{Q_i'})_i\big)}}}\\
    \multicolumn{3}{r}{\hfill\mbox{where } \forall i. \forall v_i \in \m{fv}(Q_i) . v_i \neq y_i \Rightarrow \exists j . v_i = u_j}
    \\
    \phl{\gcpres{x,
        y}{\globalaxiomp{x}{X}{y}{\dual{X}}}{\big(\axiom{x}{{X}}{w}{X}
            \pp P\big)}} 
    &\reducesto&
    \phl{P\{w/y\}}
    \\
    \phl{\gcpres{x,
        y}{\globalaxiomp{x}{X^{\bot}}{y}{{X}}}{\big(\axiom{w}{{X}}{x}{X}
        \pp P\big)}} 
    &\reducesto&
    \phl{P\{w/y\}}
  \end{array}
  \]
  \caption{Semantics for the action fragment.}
\label{fig:action_semantics}
\end{figure*}

The semantics of the action fragment of our calculus is presented in
Fig.~\ref{fig:action_semantics}. Notice that the $\beta$-reductions
are coordinated by a global type, as they correspond to multiple
parties communicating.\footnote{It may be surprising that some of the
  rules also include a restriction to a vector $\til z$, and a session
  using a vector of processes $\til S$, whose shape we do not
  inspect. This follows from the shape of coherence rules: rules such
  as $\tensor\parr$, $\oplus\with$ and $0\top$ contain an additional
  context $\Gamma$, captured here by $\til z$.}  The reduction rules
for server killing and cloning may look strange because both
\textsf{kill} and \textsf{clone} remain in the proof term after
reduction. This is because of the corresponding reduction in classical
linear logic, where it is necessary to use weakening and contraction
(corresponding to \textsf{kill} and \text{clone} respectively) also
after reduction. As a consequence, we get them as proof terms.

\section{Extending GCP with Choreographies}\label{sec:choreo}
In order to obtain full MCC, we extend the action fragment presented
in the previous section with choreography terms (interactions).
\paragraph{Syntax.} 
Unlike a process in the action fragment, a choreography, which
describes a global view of the communications of a process, will own
all of the endpoints of the sessions it describes. We call the
fragment of MCC with choreography terms the {\em interaction
  fragment}. Formally, MCC syntax is extended as follows:
\begin{displaymath}\scriptsize
  \begin{array}{rll@{\qquad}llllll}
    P ::= & \multicolumn{2}{l}{\rule{1.5mm}{0mm}\ldots \text{as in the action fragment}\ldots}
          & \mid \hl{\gengclient}   & \text{server accept/request}\\[1mm]
          & \mid \hl{\caxcut{z}{y^B}{x}{P}}     & \text{link} 
          & \mid \hl{\gengshutdown} & \text{server kill}\\[1mm]
          & \mid \hl{\gencom}       & \text{communication}
          & \mid \hl{\gengclone}    & \text{server clone}\\[1mm]
          & \mid \hl{\gencomclose}  & \text{session close}
    \\[1mm]
          & \mid \hl{\genginl}      & \text{left selection}
          & \mid \hl{\genchinlinr}  & \text{general selection}\\[1mm]
          & \mid \hl{\genginr}      & \text{right selection}
\end{array}
\end{displaymath}
The link term $\caxcut{z}{y^B}{x}{P}$ gives the choreographic view of
an axiom connected to some other process $P$ through endpoints $x$ and
$y$. A linear interaction $\gencom$ denotes a communication from
endpoints $\til x$ to the endpoint $y$, where a new session with
endpoints $\til x'$,$y'$ is created. The choreography $\gencomclose$
closes a session between endpoints $\til x$, $y$. When it comes to
branching, we have two choreographic terms denoting left and right
selection: $\genginl$ and $\genginr$. A third term, $\genchinlinr$, is
used for non-deterministic choice. In MCC, we can model non-linear
behaviour: this is done with the terms $\gengclient$, $\gengshutdown$
and $\gengclone$. The first term features a client $x$ starting a new
session with servers $\til y$, while the second term is used by
endpoint $x$ to shut down servers $\til y$. Finally, we have a term
for cloning servers so that they can be used by different clients in
different
sessions. 

\paragraph{Typing.} Fig.~\ref{fig:mcc_interaction} details the rules
for typing choreography terms. Each of these rules combines two rules
from the action fragment simulating their reduction, where the
conclusion of a rule corresponds to the redex and the premise to the
reductum. Unlike process rules, the choreography rules now also look
at $\Sigma$ to check that the interactions described conform to the
types of the ongoing sessions.  In rule~$\onebot$, we close a session
(removed from $\Sigma$) and terminate all processes involved in it.
Rule~$\tensorparr$ types the creation of a new session with protocol
$G$, created among endpoints $\til z$ and $w$; this session is stored
in $\Sigma$, while the process types are updated as in rules~$\tensor$
and $\parr$ above.  The remaining rules in the linear fragment are
similarly understood. Exponentials give rise to three rules, all of
them combining $\bang$ with another rule.  In rule~$\bangwhynot$,
process $x$ invokes the services provided by $\til y$, creating a new
session among these processes with type $G$.  Rule~$\bangw$ combines
$\bang$ with $\m{Weaken}$: here the processes providing the service
are simply removed from the context.  Finally, rule~$\bangC$ combines
$\bang$ with $\m{Contract}$, allowing a service to be
duplicated.  
\begin{figure*}[!t]
  \begin{displaymath}
  \begin{array}{l}
    \!\!\!  \textbf{Linear Fragment:}\\[2mm]
  \begin{array}{rl}
    &\infer[\Axcut]
    {\judge{\Sigma, \marka{\pair{(x,y)}{\globalaxiom{x}{A}{y}{}}}}{\markb{\genaxcut}}{\Psi \pp \seq \Gamma, \markc{\pair{x}{A}} \pp \markc{\seq \pair{w}{A}, \pair{y}{B}}}}
    {\judge{\Sigma}{\markb{P\{w/x\}}}{\Psi \pp \seq \Gamma, \markc{\pair{w}{A}}} & w \not \in \m{vars}(\Sigma) & A^\bot = B & A = X \mbox{ or } A = X^\bot }\\[5mm]
    &\infer[\tensorparr]
      {
        \judge{\Sigma,\marka{\pair{(\til x,y,\til u)}{\gcom{\til x}{y}{G};H}}}
              {\markb{\gencom}}
              {
                \Psi
                \pp
                \markc{\left( \seq \Gamma_{i1},\Gamma_{i2}, \pair{x_i}{A_i \tensort y B_i} \right)_i} \pp
                \seq \Gamma, \markc{\pair{y}{C \parrt{\til x} D}}
              }
      }
      {
        \judge{\Sigma,\marka{\pair{(\til x, y, \til u)}{H},\pair{(\til x',y')}{G\{\til x'/\til x,y'/y\}}}}
              {\markb{P}}
              {
                \Psi
                \pp
                \markc{\left( \seq \Gamma_{i1}, \pair{x'_i}{A_i} \right)_i \pp
                  \left( \seq \Gamma_{i2}, \pair{x_i}{B_i} \right)_i} \pp
                \seq \Gamma, \markc{\pair{y'}{C}, \pair{y}{D}}
              }
      }
    \\[5mm]
    &\infer[\onebot]
    {
      \judge{\Sigma, \marka{\pair{(\til x,y)}{\gclose{\til x}{y}}}}
            {\markb{\gencomclose}}
            {
              \Psi\pp
              \markc{\left( \seq \pair{x_i}{\onet y} \right)_i} \pp
              \seq \Gamma, \markc{\pair{y}{\bott{\til x}}}
            }
    }
    {
      \judge \Sigma {\markb{P}}
             {\Psi \pp \seq \Gamma}
    }
    \\[5mm]
    &\infer[\opluswith^1]
      {
        \judge{\Sigma,(\Sigma_i)_i,\marka{\pair{(x,\til y,\til u)}{\gsel{x}{\til y}GH}}}
              {\markb{\genginl}}
              {\begin{array}{l}
                \Psi \pp (\Psi_i)_i
                \pp
                \seq \Gamma, \markc{\pair{x}{A \oplust{\til y} B}}
                \pp\\
                \markc{\left( \seq \Gamma_i, \pair{y_i}{C_i \witht x D_i} \right)_i}
                \pp
                \left(\seq \Gamma_j, \pair{u_j}{E_j}\right)_j
              \end{array}}
      }
      {
	\begin{array}{l}
          \judge{\Sigma, (\Sigma_i)_i,\marka{\pair{(x,\til y,\til u)}{G}}}
	      {\markb{P}}
	      {
	        \Psi \pp \left(\Psi_i\right)_i
	        \pp
	        \seq \Gamma, \markc{\pair{x}{A}}
                \pp
                \markc{\left( \seq \Gamma_i, \pair{y_i}{C_i} \right)_i}
                \pp
                \left(\seq \Gamma_j, \pair{u_j}{E_j}\right)_j
	      }
	\\
	\left(\judge{\Sigma_i}
	      {\markb{Q_i}}
	      {
	        \Psi_i
                \pp
                \markc{\seq \Gamma_i, \pair{y_i}{D_i}}
	      }\right)_i
              \qquad\hfill
              H \gseq \pair{x}{B}, (\pair {y_i} {D_i})_i, (\pair {u_j} {E_j})_j
              \end{array}
      }
    \\[5mm]
    &\infer[\opluswith^2]
      {
        \judge{\Sigma,\left(\Sigma_i\right)_i,\marka{\pair{(x,\til y,\til u)}{\gsel{x}{\til y}GH}}}
              {\markb{\genginr}}
              {\begin{array}{l}
                \Psi \pp \left(\Psi_i\right)_i
                \pp
                \seq \Gamma, \markc{\pair{x}{A \oplust{\til y} B}}
                \pp\\
                \markc{\left( \seq \Gamma_i, \pair{y_i}{C_i \witht x D_i} \right)_i}
                \pp
                (\seq \Gamma_j, \pair{u_j}{E_j})_j              \end{array}}
      }
      {\begin{array}{l}
	\left(\judge{\Sigma_i}
	      {\markb{P_i}}
	      {
	        \Psi_i
                \pp
                \markc{\seq \Gamma_i, \pair{y_i}{C_i}}
	      }\right)_i
              \qquad\hfill
        G \gseq \pair{x}{A}, (\pair {y_i} {C_i})_i, (\pair {u_j} {E_j})_j
	\\
	\judge{\Sigma, \left(\Sigma_i\right)_i,\marka{\pair{(x,\til y,\til u)}{H}}}
	      {\markb{Q}}
	      {
	        \Psi \pp \left(\Psi_i\right)_i
	        \pp
	        \seq \Gamma, \markc{\pair{x}{B}}
                \pp
                \markc{\left( \seq \Gamma_i, \pair{y_i}{D_i} \right)_i}
                \pp
                (\seq \Gamma_j, \pair{u_j}{E_j})_j}       
        \end{array}
      }\\[5mm]
      &\infer[\opluswith]
      {\judge{\Sigma,\marka{\pair{(x,\til y,\til u)}{\gsel{x}{\til y}GH}}}
        {\markb{\genchinlinr}}
              {\Psi \pp 
                \seq \Gamma, \markc{\pair{x}{A \oplust{\til y} B}}
                \pp
                \markc{\left( \seq \Gamma_i, \pair{y_i}{C_i \witht x D_i} \right)_i}}}
      {\judge{\Sigma,\marka{\pair{(x,\til y,\til u)}{G}}}{\markb{P}}
        {\Psi \pp \seq \Gamma, \markc{\pair{x}{A}} \pp
          \markc{\left(\seq \Gamma_i, \pair{y_i}{C_i}\right)_i}} &
        \judge{\Sigma,\marka{\pair{(x,\til y,\til u)}{H}}}{\markb{Q}}
              {\Psi \pp \seq \Gamma, \markc{\pair{x}{B}} \pp
          \markc{\left(\seq \Gamma_i, \pair{y_i}{D_i}\right)_i}}}
    \\[5mm]
  \end{array}
    \\
    \!\!\!\textbf{Exponential Fragment:}\\[2mm]
  \begin{array}{rl}
    &\infer[\bangwhynot]
      {
        \judge {\Sigma, \marka{\pair{(x, \til y)}{\gbang{x}{\til y}{G}}}}
               {\markb{\gengclient}}
               {
                 \Psi \pp
                 \seq \Gamma, \markc{\pair x{\whynott{\til y} A}}
                 \pp
                 \markc{\left( \seq\  ?\Gamma_i,\  \pair{y_i}{\bangt x B_i}\  \right)_i}
               }
      }
      {
        \judge {\Sigma, \marka{\pair{(x, \til y)}G}}
               {\markb{P}}
               {
                 \Psi \pp
                 \seq \Gamma, \markc{\pair xA}
                 \pp
                 \markc{\left( \seq ?\Gamma_i, \pair{y_i}{B_i} \right)_i}
               }
        &
        \forall i \,.\,\m{vars}(?\Gamma_i) \cap \m{vars}{(\Sigma)} = \emptyset
      }
    \\[5mm]
    &\infer[\bangw]
      {
        \judge {\Sigma, \marka{\pair{(x, \til y)}{\gbang{x}{\til y}{G}}}}
               {\markb{\gengshutdown}}
               {
                 \Psi \pp
                 \seq \Gamma, \markc{\pair x{\whynott{\til y} A}}
                 \pp
                 \markc{\left( \seq\  ?\Gamma_i,\  \pair{y_i}{\bangt x B_i}\  \right)_i}
               }
      }
      {
        \judge {\Sigma} {\markb{P}}
               {
                 \Psi \pp
                 \seq \Gamma
               }
        &
        \left(\judge {\cdot}{\markb{Q_i}}
                     {
                       \markc{\seq ?\Gamma_i,\ \pair{y_i}{B_i}}
                     }
                     \right)_i
         &
         G \gseq \pair{x}{A}, \left(\pair{y_i}{B_i}\right)_i            
      }
    \\[5mm]
    &\infer[\bangC]
      {
        \judge {\Sigma, \marka{\pair{(x, \til y)}{\gbang{x}{\til y}{G}}}}
               {\markb{\gengclone}}
               {
                 \Psi
                 \pp
                 \seq \Gamma, \markc{\pair x{\whynott{\til y} A}}
                 \pp
                 \markc{\left( \seq\  ?\Gamma_i,\  \pair{y_i}{\bangt x B_i}\  \right)_i}
               }
      }
      {
        \judge 
            {\Sigma, 
              \begin{array}l
                \marka{\pair{(x, \til y)}{\gbang{x}{\til y}{G}}},\\
                \marka{\pair{(x', \til y')}{\gbang{x'}{\til y'}{G\{x'/x,\til y'/\til y\}}}}
              \end{array}
            }
            {\markb{P}}
            {
              \Psi \pp
              \begin{array}l
              \seq \Gamma, \markc{\pair {x}{\whynott{\til y} A}, \pair{x'}{\whynott{\til y'} A}} \\
              \markc{\left( \seq\  ?\Gamma_i,\  \pair{y_i}{\bangt {x} B_i} \pp \seq\  ?\Gamma_i',\ \pair{y_i'}{\bangt {x'} B_i}\ \right)_i}
              \end{array}
            }
            &
            \begin{array}{l}
              \m{vars}(?\Gamma_i) \cap \m{vars}{(\Sigma)} = \emptyset \\
              \m{vars}(?\Gamma_i') \cap \m{vars}{(\Sigma)} = \emptyset
            \end{array}
        }
  \end{array}
  \end{array}
  \end{displaymath}
  \caption{Rules for the interaction fragment.}
  \label{fig:mcc_interaction}
\end{figure*}
%

\paragraph{Reduction Semantics.} 
\begin{figure*}[!t]
  \[
  \begin{array}{rrl}
    \phl{\gcpres{x,y}{\globalaxiom{x}{X}{y}{}}{(\caxcut{w}{y^{X^\bot}}{x}{P})}}
    & \reducesto
    & \phl{P\{w/x\}}
    \\[2mm]
    \phl{\gcpres{x,y}{\globalaxiom{x}{X^\bot}{y}{}}{(\caxcut{w}{y^X}{x}{P})}}
    & \reducesto
    & \phl{P\{w/x\}}
    \\[2mm]
    \phl{\gcpres{\til x, y, \til z}
      {\gcom{\til x}{y}{G};H}
      { \left(\gencom\right) }}
    & \reducesto
    & \phl{\gcpres{\til x', y'}
      {G\{\til x'/\til x,y'/y\}} {\gcpres{\til x, y, \til z}{H}
        {P}}}
    \\[2mm]
    \phl{\gcpres{\til x, y}{\gclose{\til x}{y}} {(\gencomclose)}} & \reducesto & \phl P
    \\[2mm]
    \phl{\gcpres{x, \til y, \til z}{\gsel{x}{\til y}{G}{H}} {\left( \genginl\right)  }}
    &\reducesto&
    \phl{\gcpres{x, \til y, \til z}{G}{P}}
    \\[2mm]
    \phl{\gcpres{x, \til y, \til z}{\gsel{x}{\til y}{G}{H}} {\left( \genginr\right)  }}
    &\reducesto&
    \phl{\gcpres{x, \til y, \til z}{H}{Q}}
    \\[2mm]
    \phl{\gcpres{x, \til y, \til z}{\gsel{x}{\til y}{G}{H}} {\left( \genchinlinr\right)  }}
    &\reducesto&
    \phl{\gcpres{x, \til y, \til z}{G}{P}}
    \\[2mm]    
    \phl{\gcpres{x, \til y, \til z}{\gsel{x}{\til y}{G}{H}} {\left( \genchinlinr\right)  }}
    &\reducesto&
    \phl{\gcpres{x, \til y, \til z}{H}{Q}}
    \\[2mm]
    \phl{\gcpres{x, \til y}{\gbang{x}{\til y}{G}} { \left(\gengclient\right) }} 
    &\reducesto& \phl{\gcpres{x, \til y}GP}
    \\[2mm]
    \phl{\gcpres{x, \til y}{\gbang{x}{\til y}{G}} { \left(\gengshutdown\right) }}
    & \reducesto & \phl{\left(\srvkill{u_j}\right)_j \pp P}\ 
                   \hfill 
                   \quad(\forall v_i \in
                   \m{fv}(Q_i) . v_i \neq y_i \Rightarrow \exists j
                   . v_i = u_j)
    \\[2mm]
    \phl{\gcpres{x, \til y}{\gbang{x}{\til y}{G}} {\left(\gengclone\right)}}
    & \reducesto & \\
    \multicolumn{3}{r}{
      \phl{(\cloneo{u_j}{u_j'})_j;\gcpres{x, \til y}{\gbang{x}{\til y}{G}}
        {\gcpres{x', \til y'}{\gbang{x'}{\til y'}{G\{x'/x,\til y'/\til y\}}}
          {P}}}\hfill \mbox{(see Remark \ref{rem:clone_inter_sem})}} 
  \end{array}
  \]
        \caption{Semantics for the interaction fragment.}
  \label{fig:interaction_semantics}
\end{figure*}
%

%
Fig.~\ref{fig:interaction_semantics} gives the reductions for the
interaction fragment. From a proof-theoretical perspective, these
reductions correspond to proof transformations of $\m{C}$ rules from
Fig.~\ref{fig:mcc_interaction} followed by a structural $\m{Scope}$
rule; the transformation removes the $\m{C}$ rule and pushes
$\m{Scope}$ higher up in the proof tree.

\begin{remark}[Server Cloning]\label{rem:clone_inter_sem} 
  The reduction rule for a server cloning choreography must clone all
  of the doubled endpoints. Looking at the typing rule $\bangC$ on
  Figure \ref{fig:mcc_interaction}, cloned variables $u_j$ are all of
  the endpoints mentioned in $\left(?\Gamma_i\right)_i$, and $u_j'$
  are corresponding endpoints from $\left(?\Gamma_i'\right)_i$. To
  make the search for these variables syntactic, one could do an
  endpoint projection, as described in the next section, and look at
  the appropriate subterm of the $\m{Conn}$ rule which connects $y_i$
  and $x$. The $u_j$ are then the free variables of this subterm,
  excluding $y_i$.
\end{remark}

\paragraph{Structural equivalence.} The reductions given earlier
require that programs are written in the very specific form given in
their left-hand side.  Formally, this is achieved by closing
$\reducesto$ under structural equivalence: if $P\equiv P'$,
$P'\reducesto Q'$ and $Q'\equiv Q$, then $P\reducesto Q$. The
equivalences for the interaction part are given in
Fig.~\ref{fig:commuting}
. As in the action fragment, we are only interested in commuting
conversions of typable programs, and therefore rely on typing
derivations for finding the correct equivalence.
\begin{figure}[!t]
\[\scriptsize
\begin{array}{ll@{\!\,\qquad}l}
  \hl{\genaxcut \pp \til Q}
  \equiv \hl{\caxcut{w}{y^B}{x}{(P \pp \til Q)}} &\\[1mm]
  \hl{\gencom \pp \til Q}
  \equiv \hl{\com {\til x}{\til x'}{y}{y'}{(P \pp \til Q)}} &\\[1mm]
  \hl{\gencomclose \pp \til Q}
  \equiv \hl{\comclose {\til x}y{(P \pp \til Q)}} & \\[1mm]
  \hl{\genginl \pp \til S}
  \equiv \hl{\comsel{x}{\til y}{\lleft}{(P \pp \til S)}{Q_1,\dots,Q_n}} \\[1mm]
  \hl{\genginl \pp \til S}
  \equiv \hl{\comsel{x}{\til y}{\lleft}{(P \pp \til S)}{(Q_1, \dots, (Q_i \pp \til S), \dots,Q_n)}} \\[1mm]
  \hl{\genginr \pp \til S}
  \equiv \hl{\comsel{x}{\til y}{\lright}{P_1,\dots,P_n}{(Q \pp \til S)}}\\[1mm]
  \hl{\genginr \pp \til S}
  \equiv \hl{\comsel{x}{\til y}{\lright}{(P_1,\dots,(P_i \pp \til S),\dots,P_n)}{(Q \pp \til S)}}\\[1mm]
  \hl{\gengclient \pp \til Q}
  \equiv \hl{\gclient x{\til y}{(P\pp \til Q)}}\\[1mm]
  \hl{\gengshutdown \pp \til S}
  \equiv \hl{\gshutdown x{\wtil{y(Q)}}{(P \pp \til S)}} \\[1mm]
  \hl{\gengclone \pp \til Q}
  \equiv \hl{\gclone x{\til y}{x'}{\til y'}{(P \pp \til Q)}}\\[2mm]
    \hl {\res {{\til w}:\!G} (\gencom)} \equiv \hl {\com {\til x} {\til x'} y {y'} {\res {{\til w}:\!G}P}} & \\[1mm]
    \hl {\res {{\til w}:\!G} (\gencomclose)} \equiv \hl {\comclose {\til x}y{\res {{\til w}:\!G}P}} & \\[1mm]
    \hl {\res {{\til w}:\!G} (\genginl)}
    \equiv \hl {\comsel{x}{\til y}{\lleft}{\res {{\til w}:\!G}P}{Q_1,\dots,Q_n}}\\[1mm]
    \hl {\res {{\til w}:\!G} (\genginl)}
    \equiv \hl {\comsel{x}{\til y}{\lleft}{\res {{\til w}:\!G}P}{Q_1, \dots,\res {{\til w}:\!G}Q_i, \dots, Q_n}}\\[1mm]
      \hl {\res {{\til w}:\!G} (\genginr)}
    \equiv \hl {\comsel{x}{\til y}{\lright}{P_1,\dots,P_n}{\res {{\til w}:\!G}Q}}\\[1mm]
    \hl {\res {{\til w}:\!G} (\genginr)}
    \equiv \hl {\comsel{x}{\til y}{\lright}{P_1,\dots,\res{{\til w}:\!G}P_i,\dots,P_n}{\res {{\til w}:\!G}Q}}\\[1mm]
    \hl{\res{{\til w}:\!G} (\geninlinr)} \equiv \hl{\inlinr{x}{\res{{\til w}:\!G}P}{\res{{\til w}:\!G}Q}}\\[1mm]
    \hl {\res {{\til w}:\!G} (\gengclient)} \equiv \hl {\gclient x{\til y}{\res {{\til w}:\!G}P}} & \\[1mm]
    \hl {\res {{\til w}:\!G} (\gengshutdown)} \equiv \hl {\gshutdown x{\wtil{y_i(Q_i)}}{\res {{\til w}:\!G}P}} & \\[1mm]
   \hl {\res {{\til w}:\!G} (\gengclone)} \equiv \hl {\gclone x{\til y}{x'}{\til y'}{\res {{\til w}:\!G}P}} & \\[2mm] 
   \hl{\res{{\til z}:\!G}(\gengclone \pp \til Q)} \equiv \hl{\gclone x{\til y}{x'}{\til y'}{\res{{\til z}:\!G}(P\pp\til Q)}}
\end{array}
\]
\caption{Equivalences for commuting $\m{C}$-rules with $\m{Conn}$ and $\m{Scope}$. All rules assume
  that both sides are typable in the same context.}
\label{fig:commuting}
\end{figure}
Besides the commuting conversions, we also have the usual structural
equivalence rules where parallel composition under restriction,
linking process and global type for linking sessions are all
symmetric. Furthermore, the order of restrictions can be swapped.
\vspace{-2mm}

{\scriptsize \begin{align*}
  \hl{\axiom x A y B} & \equiv \hl{\axiom y {A^\bot} x A}
  \\
  \hl{\gcpres{\til w, y, x, \til z}{G}{\til P \pp R \pp Q \pp \til S}} 
  &\equiv \hl{\gcpres{\til w, x, y, \til z}{G}{\til P \pp Q \pp R \pp \til S}}
  \\
  \hl{\gcpres{z, \til w}{H}{\gcpres{x, \til y}{G}{P \pp \til R} \pp \til Q}}
  &\equiv 
  \phl{\gcpres{x, \til y}{G}{\gcpres{z, \til w}{H}{P \pp \til Q} \pp \til R}}
\end{align*}}

\paragraph{Properties.} We finish the presentation of MCC by
establishing the expected meta-theoretic properties of the system.
As structural congruence is typing-based, subject congruence is a
property holding by construction:
\begin{theorem}[Subject Congruence]
  $\judge \Sigma P \Psi$ and $P \equiv Q$ implies that
  $\judge \Sigma Q \Psi$.
\end{theorem}
\begin{proof}
  By induction on the proof that $P\equiv Q$.
  In~\cite{CCMM18}, 
  it is explained how the rules for structural equivalence were
  derived, making this proof straightforward.
\end{proof}
Moreover, our reductions preserve typing since they are proof
transformations.
\begin{theorem}[Subject Reduction]
  $\judge \Sigma P \Psi$ and $P \reducesto Q$ implies $\judge \Sigma Q \Psi$.
\end{theorem}
\begin{proof}
  By induction on the proof that $P\reducesto Q$.
  In~\cite{CCMM18}, 
  it is explained how the semantics of MCC were designed in order to
  make this proof straightforward.
\end{proof}
Finally, we can show that MCC is deadlock-free, since the top-level
$\m{Scope}$ application can be pushed up the derivation.  In case the
top-level $\m{Scope}$ application is next to an application of
$\m{Conn}$, either the choreography can reduce directly or both rules
can be pushed up.  Proof-theoretically, this procedure can be viewed
as MCC's equivalent of the Principal Lemma of Cut Elimination.
\begin{theorem}[Deadlock-freedom]\label{thm:deadlockfree}
  If $P$ begins with a restriction and $\judge \Sigma P \Psi$, then there exists $Q$ such that
  $P\reducesto Q$.
\end{theorem}
\begin{proof}[Sketch]
  Our proof idea is similar to that of Theorem~3
  in~\cite{CMS18}. 
  We apply induction on the size of the
  proof of $\judge\Sigma P\Psi$.  If a rule from
  Fig.~\ref{fig:action_semantics} or
  Fig.~\ref{fig:interaction_semantics} is applicable (corresponding to
  a proof where an application of $\m{Conn}$ and an application of
  $\m{Scope}$ meet), then the thesis immediately holds.

  Otherwise, we apply commuting conversions from
  Fig.~\ref{fig:commconv_action} or Fig.~\ref{fig:commuting},
  ``pushing'' the top-level $\m{Scope}$ application up in the
  derivation (and, if it is preceded by an application of $\m{Conn}$,
  ``pushing'' also that application).  This results in a smaller proof
  of $\judge\Sigma P\Psi$, to which the induction hypothesis can be
  applied.
\end{proof}



\section{Projection and Extraction}\label{sec:proj_extr}
%
As suggested by the previous sections, interactions can be implemented
in two ways: as a single choreography term, or as multiple process
terms appearing in different behaviours composed in parallel. In this
section, we formally show that choreography interactions can be
projected to process implementations, and symmetrically, process
implementations can be extracted to choreographies. We do this by
transforming proofs (derivations in the typing system), similarly to
the way we defined equivalences and reductions for MCC.

We start by defining the principal transformations for projection and
extraction, a set of equivalences that require proof terms to have a
special shape. We report such transformations in
Fig.~\ref{fig:proj_extr}: they perform extraction if read from left to
right, while they perform projection if read from right to left. The
extraction relation requires access to the list of open sessions
$\Sigma$ to ensure that we have all the endpoints participating in the
session to extract a choreography from.
%
\begin{figure}[t]
  \begin{align*}
    P \pp \axiom{x}{X}{y}{}
    & \quad\projextr\quad
    \caxcut{x}{y^{X^\bot}}{w}{P}  \;\qquad\qquad\qquad (w,y) \in \m{dom}(\Sigma)
    \\
    P \pp \axiom{y}{X}{x}{}
    & \quad\projextr\quad
    \caxcut{x}{y^X}{w}{P}  \quad\qquad\qquad\qquad (w,y) \in \m{dom}(\Sigma)
    \\
    \left(\send{x_i}{x_i'}{y',\til x'_{\setminus i}}{P_i}{Q_i}\right)_i \pp \recv{y}{\til x'}{y'}R \pp \til S
    & \quad\projextr\quad
    \com{\til x}{\til x'}{y}{y'}{\left(\til P \pp \til Q\pp R \pp \til S\right)}
    \\
    (\close{x_i})_i \pp \genwait
    & \quad\projextr\quad
    \gencomclose
    \\
    \geninl \pp (\Case{y_i}{Q_i}{R_i})_i \pp \til S
    & \quad\projextr\quad
    \comsel{x}{\til y}{\lleft}{P\pp \til Q \pp \til S }{R_1,\dots,R_n}
    \\
    \inr{x}P \pp (\Case{y_i}{Q_i}{R_i})_i \pp \til S
    & \quad\projextr\quad
    \comsel{x}{\til y}{\lright}{Q_1,\dots,Q_n}{P\pp \til R \pp \til S}
    \\
    \geninlinr \pp \left( \Case{y_i}{R_i}{S_i}\right)_i \pp \til T
    & \quad\projextr\quad
    \chinlinr{x}{\til y}{P \pp \til R \pp \til T}{Q \pp \til S \pp \til T}
    \\
    \genclient \pp (\srv{y_i}{x}{Q_i})_i
    & \quad\projextr\quad
    \gclient x{\til y}(P\pp \til Q)
    \\
    \genshutdown \pp (\srv{y_i}{x}{Q_i})_i
    & \quad\projextr\quad
    \gengshutdown
    \\
    \genclone \pp (\srv{y_i}{x}{Q_i})_i
    & \quad\projextr\quad
    \gclone x{\til y}{x'}{\til y'}\left(P \pp (\srv{y_i}{x}{Q_i})_i \pp (\srv{y'_i}{x'}{Q'_i})_i\right)
  \end{align*}
  \caption{Extraction ($\rightharpoonup$) and projection ($\leftharpoondown$).}
  \label{fig:proj_extr}
\end{figure}
%
%

%
The first two rules deal with axioms: the parallel composition (rule
$\m{Conn}$) of an axiom with a process $P$ can be expressed by rule
$\Axcut$ and vice-versa. On the third line, we show how to transform
the parallel composition of an output ($\tensor$) and an input
($\parr$) into a $\tensorparr$. Similarly, $\gencomclose$ is the
choreographic representation of the term
$(\close{x_i})_i \pp \genwait$. Each branching operation (left, right,
non-deterministic) has a representative in both fragments with
straightforward transformations. A server $\srv{y}{x}{Q}$ can either
be used by a client, killed or cloned. In the first two cases, such
interactions trivially correspond to the choreographic terms
$\gclient x{\til y}(P\pp \til Q)$ and $\gengshutdown$. In the case of
cloning, we create the interaction term
$ \gclone x{\til y}{x'}{\til y'}\left(P \pp (\srv{y_i}{x}{Q_i})_i \pp
  (\srv{y'_i}{x'}{Q'_i})_i\right)$, which shows how the choreographic
cloning $\gclone x{\til y}{x'}{\til y'}$ must be followed by two
instances of the server that is cloned.
Note that these transformations are derived by applying similar
techniques as those of cut elimination.
Concrete derivations, here omitted, are straightforward: an example can be
found in~\cite{CCMM18}. 

\begin{remark}
  \label{rem:proj_extr}
  In order to project/extract an arbitrary well-typed term, given the
  strict format required by the transformations in
  Fig.~\ref{fig:proj_extr}, we will sometimes have to perform
  rewriting of terms in accordance with the commuting conversions to
  reach an expected shape. In particular, we note that when
  projecting, we must first project the subterms (we start from the
  leaves of a proof), step by step moving down to the main term. In
  contrast, when extracting, we must proceed from the root of the
  proof towards the leaves.
\end{remark}

Note that our example in \S\ref{sec:preview} does not provide an exact
projection: in order to improve readability, we have removed all
parallels that follow output operations, which would be introduced by
the translation presented above. This is not problematic, since the
outputs in the example are just basic types.

\paragraph{Properties.} In the sequel, we write
$P\stackrel{\til x}{\longrightarrow}_{\m{extr}} P'$ whenever it is
possible to apply one of the transformations in
Fig.~\ref{fig:proj_extr} to (a term equivalent to) term $P$ from left
to right, where $\til x$ are the endpoints involved in the
transformation. Similarly, we write
$P\stackrel{\til x}{\longrightarrow}_{\m{proj}} P'$ whenever it is
possible to apply a transformation from Fig.~\ref{fig:proj_extr} to (a
term equivalent to) term $P$ from right to left. We also write
$P\Longrightarrow_{\m{extr}} P'$ ($P\Longrightarrow_{\m{proj}} P'$) if
there is a finite sequence of applications of
$\longrightarrow_{\m{extr}}$ ($\longrightarrow_{\m{proj}}$) and $P'$
cannot be further transformed. We then have the following results:
\begin{theorem}[Type Preservation]\label{thm:projextr_preserv}
  If $P\stackrel{\til x}{\longrightarrow}_{\m{extr}}Q$ and $\judge\Sigma P\Psi$, then
  $\judge\Sigma Q\Psi$, and if $Q\stackrel{\til x}{\longrightarrow}_{\m{proj}}P$ and
  $\judge\Sigma Q\Psi$, then $\judge\Sigma P\Psi$.
\end{theorem}
\begin{proof}
  By induction on the proof that $P\stackrel{\til
    x}{\longrightarrow}_{\m{extr}} Q$ or $Q\stackrel{\til
    x}{\longrightarrow}_{\m{proj}} P$.  In~\cite{CCMM18},
  we 
  explain how the rules for projection and extraction were derived
  from the typing rules to ensure that the proof of this result is
  straightforward.
\end{proof}
\begin{theorem}[Admissibility of $\m{Conn}$ and $\m{C}$-rules] Let $P$
  be a proof term such that $\judge{}{P}{\seq \Gamma}$. Then,
  \begin{itemize}
  \item there exists $P'$ such that $P\Longrightarrow_{\m{extr}} P'$
    and $P'$ is $\m{Conn}$-free;

  \item there exists $P'$ such that $P\Longrightarrow_{\m{proj}} P'$
    and $P'$ is free from $\m{C}$-rules.
  \end{itemize}
\end{theorem}
\begin{proof}[Sketch]
  The idea is similar to the proof of Theorem~4.4.1 in~\cite{M13:phd}:
  by applying commuting conversions we can always rewrite $P$ such
  that one of the rules in Fig.~\ref{fig:proj_extr} is applicable,
  thus eliminating the outermost application of $\m{Conn}$ (in the
  case of extraction) or the innermost application of a $\m{C}$-rule
  (in the case of projection).  See also Remarks~\ref{rem:proj_extr}
  and~\ref{rem:proj_nocontext}.
\end{proof}
\begin{remark}
  \label{rem:proj_nocontext}
  The theorem above is only applicable to judgments of the form
  $\judge{}{P}{\seq \Gamma}$. 
  This is because of the commuting conversion of the server rule
  \[\res{{\til x
    x}:\!G}(\gensrv \pp \til Q) \equiv \srv{y}{}{\res{{\til x
      x}:\!G}(P \pp \til Q)}\]
where we can only permute $\m{Conn}$ and $\m{Scope}$ together. This
conversion is needed to rearrange certain proofs into the format
required by the transformations in Fig.~\ref{fig:proj_extr}.  Note
that any judgement $\judge{\Sigma}{P}{\Psi}$ can always be transformed
into this format, by repeatedly applying rule \m{Scope} to all
elements in $\Sigma$.
\end{remark}

As a consequence of the admissibility of \m{Conn}, every program can
be rewritten into a (non-unique) process containing only process terms
by applying the rules in Fig.~\ref{fig:proj_extr} from right to left
until no longer possible.  Conversely, because of admissibility of
$\m{C}$-rules, every program can be rewritten into a maximal
choreographic form by applying the same rules from left to right until
no longer possible. 

We conclude this section with our main theorem that shows the
correspondence between the two fragments with respect to their
semantics. In order to do that, we annotate our semantics with the
endpoints where the reduction takes place. This is denoted by
$P\reducesto^{\til x} Q$ and $P\reducesto^{\bullet\til x} Q$ where
the first relation is a reduction in the action fragment, while the
second is a reduction in the interaction fragment.
The sequence $\m{rev}(\til x)$ is obtained by reversing $\til x$.
\begin{theorem}[Correspondence]
 Let $P$ be a proof term such that $\judge{\Sigma}{P}{\Psi}$. Then,
 \begin{itemize}
 \item $P\reducesto^{\til x} Q$ implies that there exists $P'$ s.t.
   $P\stackrel{\til x}{\longrightarrow}_{\m{extr}} P'$ and
   $P'\reducesto^{\bullet\til x} Q$;
 \item $P\reducesto^{\bullet\til x} Q$ implies that there exists $P'$ s.t.
   $P\stackrel{\m{rev}(\til x)}{\longrightarrow}_{\m{proj}} P'$ and
   $P'\reducesto^{\til x} Q$.
 \end{itemize}
\end{theorem}
\begin{proof}
  This proof follows the same strategy as that of Theorem~6 in~\cite{CMS18}.
\end{proof}

\section{Related Work and Discussion}\label{sec:related}
\paragraph{Related Work.}
The principle of choreographies as cut reductions was introduced in
\cite{CMS18}. As discussed in \S\ref{sec:intro}, that system cannot
capture services or multiparty sessions.  Another difference is that
it is based on intuitionistic linear logic, whereas ours on classical
linear logic -- in particular, on Classical Processes~\cite{W14}.

Switching to classical linear logic is not a mere change of
appearance. It is what allows us to reuse the logical understanding of
multiparty sessions in linear logic as \emph{coherence proofs},
introduced in \cite{CMSY17} and later extended to polymorphism in
\cite{CLMSW16}. These works did not consider choreographic programs,
and thus do not offer a global view on how different sessions are
composed, as we do in this paper.

Extracting choreographies from compositions of process code is
well-known to be a hard problem.  In \cite{LTY15}, choreographies that
abstract from the exchanged values and computation are extracted from
communicating finite-state machines. The authors of \cite{CLM17}
present an efficient algorithm for extracting concrete choreographic
programs with asynchronous messaging.  These works do not consider the
composition of multiple sessions, multiparty sessions, and services, 
as in MCC. However, they can both deal with infinite
behaviour (through loops or recursion), which we do not address. An
interesting direction for this feature would be to integrate
structural recursion for classical linear logic \cite{LM16}.

Our approach can be seen as a principled reconstruction of previous
works on choreographic programming. The first work that typed
choreographies using multiparty session types is \cite{CM13}. The idea
of mixing choreographies with processes using multiparty session types
is from \cite{MY13}. None of these consider extraction. 

\paragraph{Discussion.}
\noindent For the sake of clarity, our presentation of MCC adopts
simplifications that may limit the model expressivity. Below, we
discuss some key points as well as possible extensions based on
certain developments in this research line.

\mypar{Non-determinism} We introduced non-determinism in a
straightforward way, i.e., our non-deterministic rules in both action
and interaction fragments require for each branch to have the same
type, as done for standard session typing. However, this solution
breaks the property of confluence that we commonly have in logics. In
order to preserve confluence, we would have to extend MCC with the
non-deterministic linear types from \cite{CP17}.

\mypar{$\eta$-expansion} GCP in~\cite{CLMSW16} allows for the axiom to
be of any type $A$. This requires heavily using $\eta$-expansions for
transforming axioms into processes with communication actions. It is
straightforward to do this in the action fragment of MCC. However,
given the way choreographies work, we can only define an axiom for
binary sessions in the interaction fragment. As a consequence, in
order to apply extraction to a process where an axiom is engaged in a
multiparty session with several endpoints, we would need to first use
$\eta$-expansions to transform such axiom into an ordinary process. In
the opposite direction, we would never be able to project a process
containing an axiom from a choreography, unless it is part of a binary
session. We leave further investigation of this as future work.

\mypar{Annotated Types} The original version of GCP~\cite{CLMSW16}
comes with an extension called MCP, where an endpoint type $A$ is
annotated with names of endpoints which it will be in a session
with. In this way, endpoint types become more expressive, since it is
possible to specify with whom each endpoint has to communicate,
without having to use a global type (coherence proof) during
execution. We claim that this extension is straightforward for our
presentation of MCC.

\mypar{Polymorphism} As in GCP~\cite{CLMSW16}, we can easily add
polymorphic types to MCC. However, for simplifying the presentation of
this work, we have decided to leave it out, even though adding the GCP
rules to the action fragment is straightforward. In the case of the
interaction fragment, we obtain the following rule:
\begin{displaymath}\scriptsize
  \infer[\existsforall]
  {
    \judge {\Sigma, \marka{\pair{(x, \til y, \til u)}{\gexistsforall x{\til y}XG}}}
    {\markb{\gencomt}}
    {
      \Psi \pp
      \seq \Gamma, \markc{\pair x{\exists X.B}}
      \pp
      \markc{\left( \seq \Gamma_i, \pair{y_i}{\forall X.B_i} \right)_i}}
  }
  {
    \begin{array}{l}
      X \not \in \m{fv}(\Psi, \Gamma, (\Gamma_i)_i) \\[2mm]
      \judge {\Sigma, \marka{\pair{(x, \til y, \til u)}{G\{A/X\}}}}
      {\markb{P\{A/X\}}}
      {
      \Psi \pp
      \seq \Gamma, \markc{\pair x{B\{A/X\}}}
      \pp
      \markc{\left( \seq \Gamma_i, \pair{y_i}{B_i\{A/X\}} \right)_i}
      }
    \end{array}
  }
\end{displaymath}
Above we have added to the syntax of global types the term
$\gexistsforall x{\til y}XG$, denoting a session where an endpoint $x$
is supposed to send a type to endpoints $\til y$. At choreography
level, endpoint $x$ realises the abstraction of the global type
sending the actual type $A$. When it comes to extraction and
projection, we would have to add the following transformation:
\[ \gensendt \pp (\recvtype{y_i}x{X}{Q_i})_i \pp \til S 
  \quad\projextr\quad \comtype x{\til y}AX{\left(P\pp \til Q \pp \til
      S\right)}
 \]
 where $\gensendt$ and $\recvtype{y_i}x{X}{Q_i}$ are action fragment
 terms (as those of GCP). 
%

\mypar{Other Extensions} By importing the functional stratification
from \cite{TCP13}, we could obtain a monadic integration of
choreographies with functions. The calculus of classical higher-order
processes \cite{M17} could be of inspiration for adding code mobility
to MCC, by adding higher-order types.  Types for manifest sharing in
\cite{BP17} may lead us to global specifications of sharing in
choreographies.  And the asynchronous interpretation of cut reductions
in \cite{DCPT12} might give us an asynchronous implementation of
choreographies in MCC. We leave an exploration of these extensions to
future work. Hopefully, the shared foundations of linear logic will
make it possible to build on these pre-existing technical developments
following the same idea of choreographies as cut reductions.


\bibliographystyle{abbrv}
\bibliography{biblio}


\end{document}